\documentclass[reqno]{amsart}
\usepackage{amssymb,epsfig,graphics,graphicx,bbm,hyperref,color}
\usepackage{pgf,pgfarrows,pgfnodes,pgfautomata,pgfheaps,pgfshade}

\usepackage{multicol}
\usepackage{enumitem}
\usepackage{tikz-cd}
\usetikzlibrary{matrix,arrows,decorations.pathmorphing,positioning}
\usepackage{nicefrac}

\definecolor{gray}{rgb}{0.93,0.93,0.93}
\definecolor{light-gold}{rgb}{0.99,0.97,0.78}

\definecolor{gold}{rgb}{0.7,0.55,0}

\def\be{\begin{equation}}
\def\ee{\end{equation}}
\def\bm{\begin{multline}}
\def\bfig{\begin{figure}[htb]}
\def\efig{\end{figure}}
\newcommand{\dd}{{\rm d}}
\newcommand{\e}[1]{\,{\rm e}^{#1}\,}
\newcommand{\ii}{{\rm i}}

\def\Tr{{\operatorname{Tr\,}}}
\def\artanh{{\operatorname{artanh\,}}}
\def\supp{{\rm{supp\,}}}

\newcommand{\scr}[1]{^{\scriptscriptstyle (#1)}}

\numberwithin{equation}{section}
\newtheorem{theorem}{Theorem}[section]
\newtheorem{proposition}[theorem]{Proposition}
\newtheorem{lemma}[theorem]{Lemma}
\newtheorem{corollary}[theorem]{Corollary}

\newcommand{\eps}{{\varepsilon}}
\newcommand{\bbC}{{\mathbb C}}

\newcommand{\bbN}{{\mathbb N}}

\newcommand{\bbR}{{\mathbb R}}

\newcommand{\bbT}{{\mathbb T}}
\newcommand{\bbZ}{{\mathbb Z}}

\newcommand{\caE}{{\mathcal E}}

\newcommand{\caH}{{\mathcal H}}
\newcommand{\caN}{{\mathcal N}}

\newcommand{\sff}{{\mathsf f}}

\newcommand{\frah}{{\mathfrak h}}

\makeatletter
  \def\tagform@#1{\maketag@@@{\scriptsize{(#1)}\@@italiccorr}}
\makeatother

\renewcommand{\eqref}[1]{(\ref{#1})}

\begin{document}


\title[Kac-Ward solution of 2D classical and 1D quantum Ising models]{Kac-Ward solution of the 2D classical and 1D quantum Ising models}
 
\author{Georgios Athanasopoulos}
\address{Department of Mathematics, University of Warwick,
Coventry, CV4 7AL, United Kingdom}
\email{Georgios.Athanasopoulos@warwick.ac.uk}
 
\author{Daniel Ueltschi}
\address{Department of Mathematics, University of Warwick,
Coventry, CV4 7AL, United Kingdom}
\email{daniel@ueltschi.org}

\subjclass{82B05; 82B20; 82B23}


\begin{abstract}
We give a rigorous derivation of the free energy of (i) the classical Ising model on the triangular lattice with translation-invariant coupling constants, and (ii) the one-dimensional quantum Ising model.  We use the method of Kac and Ward.
The novel aspect is that the coupling constants may have negative signs.
We describe the logarithmic singularity of the specific heat of the classical model and the validity of the Cimasoni--Duminil-Copin--Li formula for the critical temperature. We also discuss the quantum phase transition of the quantum model.
\end{abstract}

\thanks{\copyright{} 2024 by the authors. This paper may be reproduced, in its entirety, for non-commercial purposes.}

\maketitle

\section{Introduction}

Onsager's calculation in 1944 of the free energy of the Ising model on the square lattice was a remarkable achievement \cite{Ons}. It helped to characterise the nature of the phase transition and yielded some critical exponents. Onsager's method was algebraic in nature and was simplified by Kaufman \cite{Kau}. The formula for the Ising free energy on the triangular lattice was first found by Houtappel \cite{Hou} in 1950; he used a simplified version of Kaufman's method with more elementary group theory. Further works on the triangular lattice (or its dual, the hexagonal lattice) include Wannier \cite{Wan}, and Husimi and Syozi \cite{HS1,HS2}.

After the work of Onsager and Kaufman, people found two alternate approaches: combinatorial and fermionic. The former was proposed in 1952 by Kac and Ward \cite{KW}; it was later extended by Kasteleyn who noted the connection with dimer systems  \cite{Kas} (see also Temperley and Fisher \cite{TF}). Potts \cite{Pot} and Stephenson \cite{Ste} used the Kac-Ward method on the triangular lattice, for the free energy and for correlation functions. The fermionic method was proposed in 1964 by Schultz, Mattis, and Lieb \cite{SML}.

In this article we use the Kac-Ward approach. It consists of two parts. First is a remarkable identity that relates the partition function of the Ising model to (the square-root of) the determinant of a suitable matrix; this holds for arbitrary planar graphs. Second, one uses the Fourier transform to block-diagonalise the matrix so as to obtain its determinant. The latter step involves a ``mild" modification of the matrix to make it periodic; this mild step has been used over the years without mathematical justification. Only recently, careful analyses have been proposed by Kager, Lis, and Meester \cite{KLM} (see \cite{Lis} for a clear description) and by Aizenman and Warzel \cite{AW} (who elucidate the connection to the graph zeta function). These analyses are restricted to nonnegative coupling constants. Another line of research is the determination of the critical temperature for general two-periodic planar graphs by Li \cite{Li} and Cimasoni and Duminil-Copin \cite{CDC}; this uses the results of Kenyon, Okounkov and Sheffield \cite{KOS} for dimer systems.

The main goal of this article is to extend the Kac-Ward method to the case of (translation-invariant) coupling constants of arbitrary signs. We work on the triangular lattice, which is the simplest case of frustrated systems with translation-invariant coupling constants. We start with the Cimasoni extension of the Kac-Ward formula to ``faithful projections" of non-planar graphs \cite{Cim} (see also Aizenman and Warzel \cite{AW} for a clear exposition). We use it for the torus $\{1,\dots,L\}_{\rm per} \times \{1,\dots,M\}_{\rm per}$ with periodic boundary conditions. The main difficulties involve the non-planarity of the graph. We prove that these difficulties vanish in the limit $L\to\infty$ for fixed $M$. Then we can use the Fourier transform and we obtain the free energy formula for the infinite cylinder $\bbZ \times \{1,\dots,M\}_{\rm per}$. The Onsager-Houtappel formula immediately follows by taking the limit $M\to\infty$.

As is well-known, the exact form of the free energy allows to establish the occurrence of a phase transition characterised by the divergence of the specific heat (the second derivative of the free energy with respect to the temperature). We discuss cases where this phase transition occurs, or fails to occur.

Our result for cylinders allows us to consider the one-dimensional quantum Ising model, whose free energy was first calculated in 1970 by Pfeuty \cite{Pfe}. We refer to \cite{Gri,BG,Iof,CI,Bjo,Li2,Tas} for recent studies. The quantum Ising model can be mapped to a 2D classical Ising model in the limit where the extra dimension becomes continuous. We also discuss the occurrence of a ``quantum phase transition".
 
The paper is organised as follows:
We state our main theorem about the free energy of the Ising model on triangular lattices in Section \ref{sec free energy}. We then discuss the possibility of a phase transition in the form of logarithmic singularity of the specific heat in Section \ref{sec log singularity}. In Section \ref{sec special case} we consider the special case where two coupling constants are equal; we show that the Cimasoni--Duminil-Copin--Li formula (see Eq.\ \eqref{CDCL}) may yield the correct critical temperature even when the couplings are not all positive. The derivation of the free energy is described in Section \ref{sec Kac-Ward}. The quantum Ising model is discussed in Section \ref{sec qu Ising}; we describe the quantum phase transition at the end of the section.

\section{The classical Ising model on the triangular lattice}
\label{sec class Ising}

\subsection{The free energy}
\label{sec free energy}

We view the triangular lattice as a square lattice with additional North-East edges.
Let $L,M \in \bbN$. Let $\bbT_L$ be the torus of $L$ sites, $\bbT_L \simeq \bbZ \setminus L \bbZ$, and let $\bbT_{L,M}$ the two-dimensional torus
\be
\bbT_{L,M} = \bbT_L \times \bbT_M.
\ee
We let $\caE_{L,M} = \caE_{L,M}^{\rm hor} \cup \caE_{L,M}^{\rm ver} \cup \caE_{L,M}^{\rm obl}$ denote the set of edges of $\bbT_{L,M}$ where
\begin{eqnarray*}
&\caE_{L,M}^{\rm hor} = \bigl\{ \{x,x+e_1\} : x \in \bbT_{L,M} \bigr\} \qquad &\text{(horizontal edges)} \\
&\caE_{L,M}^{\rm ver} = \bigl\{ \{x,x+e_2\} : x \in \bbT_{L,M} \bigr\} \qquad &\text{(vertical edges)} \\
&\caE_{L,M}^{\rm obl} = \bigl\{ \{x,x+e_1+e_2\} : x \in \bbT_{L,M} \bigr\} \qquad &\text{(oblique North-East edges)}
\end{eqnarray*}
This is illustrated in Fig.\ \ref{fig lattice}.
Let $J_1,J_2,J_3 \in \bbR$ be three parameters; we define the coupling constants $(J_e)_{e \in \caE_{L,M}}$ to be
\be
\label{def couplings}
J_e = \begin{cases} J_1 & \text{if } e \in \caE_{L,M}^{\rm hor}, \\ J_2 & \text{if } e \in \caE_{L,M}^{\rm ver}, \\ J_3 & \text{if } e \in \caE_{L,M}^{\rm obl}. \end{cases}
\ee

\begin{figure}[htb]\center
\includegraphics[width=9cm]{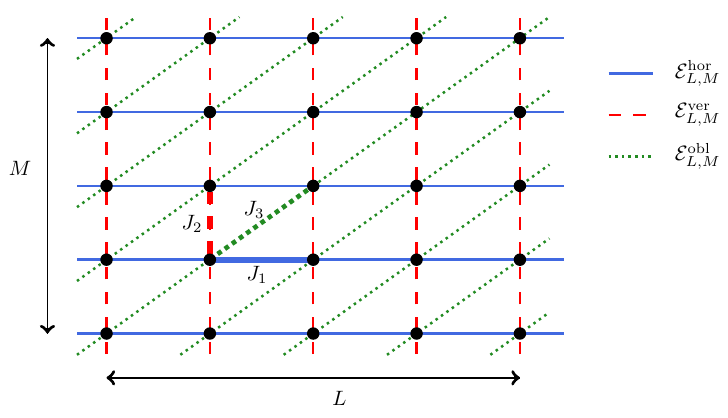}
\caption{Our lattice is the torus $\bbT_{L,M}$ with horizontal, vertical, and North-East edges.}
\label{fig lattice}
\end{figure}

A spin configuration $\sigma$ is an assignment of a classical spin $\pm1$ to each site of $\bbT_{L,M}$, $\sigma = (\sigma_x)_{x \in \bbT_{L,M}} \in \{-1,+1\}^{\bbT_{L,M}}$. The Ising hamiltonian is the function of spin configurations given by
\be
H_{L,M}(\sigma) =  -\sum_{e = \{x,y\} \in \caE_{L,M}} J_e \sigma_x \sigma_y.
\ee
The partition function is
\be
\label{def part fct}
Z_{L,M}(J_1,J_2,J_3) = \sum_\sigma \e{-H_{L,M}(\sigma)}
\ee
and the finite-volume free energy density is
\be
f_{L,M}(J_1,J_2,J_3) = -\frac1{LM} \log Z_{L,M}(J_1,J_2,J_3).
\ee

We consider two infinite-volume limits, to the infinite cylinder and to the plane. Namely, we define
\be
\begin{split}
f_M(J_1,J_2,J_3) &= \lim_{L\to\infty}  f_{L,M}(J_1,J_2,J_3); \\
f(J_1,J_2,J_3) &= \lim_{L\to\infty}  f_{L,L}(J_1,J_2,J_3).
\end{split}
\ee
As is well-known we can consider arbitrary van Hove sequences of increasing domains, see e.g.\ \cite{FV}, and we also get $f(J_1,J_2,J_3)$. The next theorem gives the free energy for the cylinder and for the two-dimensional lattice. 
The cylinder formula turns out to be convenient and it is useful in the calculation of the 1D quantum Ising model.

\begin{theorem}
\label{thm triangle Ising}
For any $J_1, J_2, J_3 \in \bbR$ we have (with $k_3 = k_1 + k_2$):
\begin{itemize}
\item[(a)] On the cylinder $\bbZ \times \bbT_M$:
\begin{multline*}
f_M(J_1,J_2,J_3) =  -\log2 - \frac1{4\pi M} \int_{-\pi}^{\pi} \dd k_1 \sum_{k_2 \in \widetilde{\bbT}_M} \log \biggl[ \prod_{i=1}^3 \cosh(2J_i)  + \prod_{i=1}^3 \sinh(2J_i) \\
- \sum_{i=1}^3 \sinh(2J_i) \cos k_i \biggr]
\end{multline*}
where $\widetilde{\bbT}_M =\frac{2\pi}{M} \bbT_M + \frac{\pi}{M}$.
\item[(b)] On the square or triangular lattice:
\begin{multline*}
f(J_1,J_2,J_3) = 
 -\log2 - \frac1{8\pi^2} \int_{[-\pi,\pi]^2} \dd k_1 \dd k_2 \log \biggl[ \prod_{i=1}^3 \cosh(2J_i) + \prod_{i=1}^3 \sinh(2J_i) \\
 - \sum_{i=1}^3 \sinh(2J_i) \cos k_i \biggr].
\end{multline*}
\end{itemize}
\end{theorem}

Setting $J_3=0$ and $J_1=J_2=J$  we get Onsager's formula for the isotropic Ising model  on the square lattice, namely
\be
f(J,J,0)=-\log2 - \frac{1}{8\pi^2} \int_{[0,2\pi]^2} \dd k_1 \dd k_2 \log\Big[\cosh^2(2J)-\sinh(2J) (\cos k_1 + \cos k_2 )\Big].
\ee

The proof of part (a) of the theorem can be found at the end of Section \ref{sec Kac-Ward}.
The next lemma establishes that $f$ is equal to the limit $M\to\infty$ of $f_M$ so that (b) immediately follows from (a).

\begin{lemma}
As $M\to\infty$ the cylinder free energy density converges to the two-dimensional free energy density:
\[
f(J_1,J_2,J_3) = \lim_{M\to\infty} f_M(J_1,J_2,J_3).
\]
\end{lemma}

\begin{proof}
We omit the dependence on coupling constants to alleviate the notation. Let $J_0 = \max_{i=1,2,3} |J_i|$.
Writing $L = kM+R$ with $R \in \{0,M-1\}$ we have
\be
Z_{M,M}^k \e{-4 J_0 kM - 6 J_0 RM} \leq Z_{L,M} \leq Z_{M,M}^k \e{4 J_0 kM + 6 J_0 RM}.
\ee
Taking the logarithm and dividing by $LM$ we get
\be
\tfrac{kM}L f_{M,M} + \tfrac{4 J_0 k + 6 J_0 R}L \geq f_{L,M} \geq \tfrac{kM}L f_{M,M} - \tfrac{4 J_0 k + 6 J_0 R}L.
\ee
We take the limit $L\to\infty$; since $kM/L \to 1$, $k/L \to 1/M$, and $R/L \to 0$ we obtain
\be
f_{M,M} + \tfrac{4 J_0}M \geq \lim_{L\to\infty} f_{L,M} \geq f_{M,M} - \tfrac{4 J_0}M.
\ee
The lemma follows by taking the limit $M\to\infty$.
\end{proof}

\begin{figure}[htb]\center
\includegraphics[width=6cm]{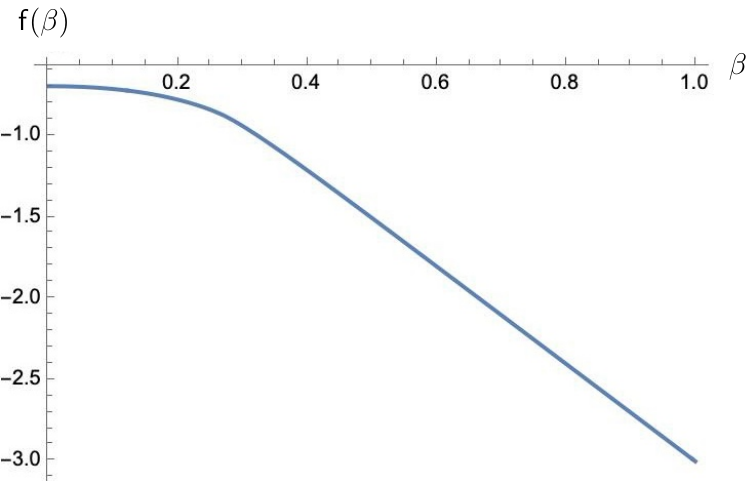}

\includegraphics[width=6cm]{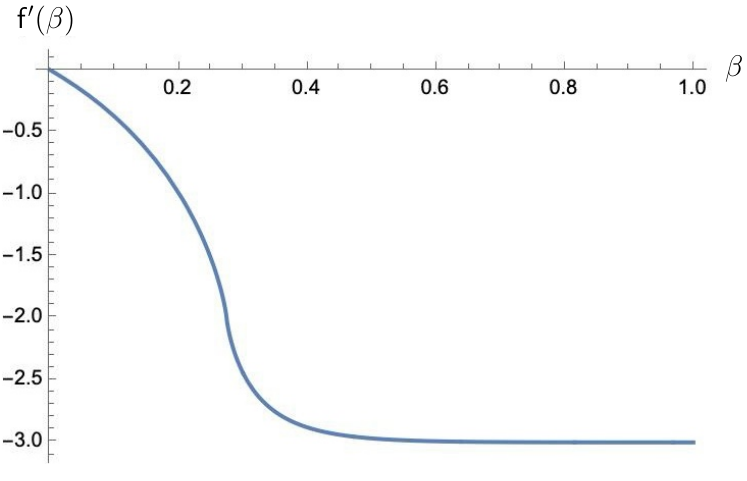}

\includegraphics[width=6cm]{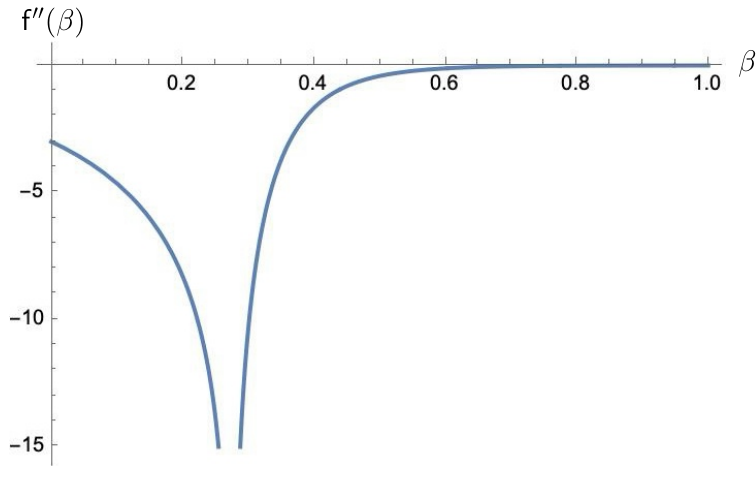}
\caption{Plots of the free energy $\mathsf f(\beta)$ and its first and second derivatives for the translation-invariant triangular lattice ($J_1=J_2=J_3=1$). The second derivative has a logarithmic singularity at $\beta_{\rm c} = \frac14 \log3 = 0.274...$.}
\label{fig plots}
\end{figure}

\subsection{Logarithmic singularity of the specific heat}
\label{sec log singularity}

We explore the consequences of the formula of Theorem \ref{thm triangle Ising} (b) regarding the possibility of phase transitions. More specifically, given fixed parameters $J_1, J_2, J_3$, we consider the function $\sff : \bbR_+ \to \bbR$:
\be
\sff(\beta) = f(\beta J_1, \beta J_2, \beta J_3).
\ee
We are looking for values of $\beta$ where $\sff$ is not analytic. We show the well-known fact that the second derivative of $\sff$ (which is related to the physical quantity called the specific heat) has a logarithmic singularity at a special value $\beta_{\rm c}$, called the critical point. This is illustrated in Fig.\ \ref{fig plots}, which displays the free energy $\sff(\beta)$ and its first and second derivatives in the case of the homogenous triangular lattice ($J_1=J_2=J_3=1$).
By Theorem \ref{thm triangle Ising} (b) we have
\be
\label{free energy 2}
\sff(\beta) = -\log 2 - \frac1{8\pi^2} \int_{[-\pi,\pi]^2} \dd k_1 \dd k_2 \, \log \bigl[ g(\beta) + h(\beta;k_1,k_2) \bigr],
\ee
where (recalling that $k_3 = k_1+k_2$)
\be
\begin{split}
&g(\beta) = \prod_{i=1}^3 \cosh(2\beta J_i) + \prod_{i=1}^3 \sinh(2\beta J_i) - \sum_{i=1}^3 \sinh(2\beta J_i), \\
&h(\beta; k_1, k_2) = \sum_{i=1}^3 \sinh(2\beta J_i) \, (1 - \cos k_i).
\end{split}
\ee
It turns out that the term inside the logarithm is always positive.

\begin{lemma}
\label{lem simply confusing}
For all $J_1, J_2, J_3 \in \bbR$, all $\beta>0$, and all $k_1, k_2 \in [-\pi,\pi]$, we have
\[
g(\beta) + h(\beta;k_1,k_2) \geq 0.
\]
\end{lemma}

There should be a simple direct proof for this lemma but we could not find one (in the case where $J_1=J_2$, it follows from the proof of Theorem \ref{thm special case} below). Instead we obtain it in Section \ref{sec Kac-Ward} using suitable Kac-Ward identities, see Corollary \ref{Cor} (a).
We now give a criterion for the free energy to be analytic in $\beta$.

\begin{lemma}
Assume that $g(\beta_0) + h(\beta_0;k_1,k_2) > 0$ for all $k_1,k_2 \in [-\pi,\pi]$. Then $\sff(\beta)$ is analytic in a complex neighbourhood of $\beta_0$.
\end{lemma}

\begin{proof}
This is a standard complex analysis argument. There exists a complex neighbourhood $\caN$ of $\beta_0$ such that $\log [g(\beta) + h(\beta;k_1,k_2)]$ is analytic in $\beta$ for each $k_1,k_2$. Then $\int_\gamma \log [g(\beta) + h(\beta;k_1,k_2)] \dd\beta = 0$ for any contour $\gamma$ in $\caN$. By Fubini's theorem,
\bm
\int_\gamma \dd\beta \int_{[-\pi,\pi]^2} \dd k_1 \dd k_2 \log [g(\beta) + h(\beta;k_1,k_2)] \\
= \int_{[-\pi,\pi]^2} \dd k_1 \dd k_2  \int_\gamma \dd\beta \log [g(\beta) + h(\beta;k_1,k_2)] = 0,
\end{multline}
so that $\sff(\beta)$ is indeed analytic in $\caN$.
\end{proof}

Next we establish a sufficient criterion for the logarithmic divergence of $\sff''(\beta)$. We assume here that the minimum of $h(\beta_{\rm c};k_1,k_2)$ is at $k_1 = k_2 = 0$ where this function is 0.

\begin{proposition}
\label{prop singularity}
Assume that there exists $\beta_{\rm c}$ such that
\[
g(\beta_{\rm c}) = 0, \quad g''(\beta_c) > 0.
\]
Further, we assume that there exists $c>0$ such that for all $k_1,k_2 \in [-\pi,\pi]$,
\[
h(\beta_{\rm c}; k_1,k_2) \geq c (k_1^2+k_2^2).
\]
Then $\sff$ is continuously differentiable at $\beta_{\rm c}$, but its second derivative diverges as $\log|\beta-\beta_{\rm c}|$ when $\beta$ approaches $\beta_{\rm c}$.
\end{proposition}

It is not hard to verify that the second condition holds true when $\sinh(2\beta_{\rm c} J_i) + 2 \sinh(2\beta_{\rm c} J_3) > 0$ for $i = 1,2$.

\begin{proof}
We already know that $\sff(\beta)$ is concave and therefore continuous. For $\beta \neq \beta_{\rm c}$, we have
\be
\label{derivative of f}
\sff'(\beta) = -\frac1{8\pi^2} \int \dd k_1 \dd k_2 \frac{g'(\beta) + \frac\partial{\partial\beta} h(\beta;k_1,k_2)}{g(\beta) + h(\beta;k_1,k_2)}.
\ee
There exists a constant $C$ such that
\be
\label{bound for ratio}
\biggl| \frac{g'(\beta) + \frac\partial{\partial\beta} h(\beta;k_1,k_2)}{g(\beta) + h(\beta;k_1,k_2)} \biggr| \leq C \frac{|\beta-\beta_{\rm c}| + k_1^2 + k_2^2}{(\beta-\beta_{\rm c})^2 + c(k_1^2 + k_2^2)}.
\ee
As $a \to 0+$, we note that
\be
\label{log divergence}
\int_0^1 \frac{r \dd r}{a^2 + r^2} = \tfrac12 \log(a^2+1) - \log a \sim |\log a|.
\ee
Writing the integral \eqref{derivative of f} with polar coordinates around 0, and using \eqref{bound for ratio} and \eqref{log divergence}, we easily check that $\sff'$ is continuous at $\beta_{\rm c}$. For the second derivative we write
\bm
\label{second derivative of f}
\sff''(\beta) = -\frac{g''(\beta)}{8\pi^2} \int \frac{\dd k_1 \dd k_2}{g(\beta) + h(\beta;k_1,k_2)} - \frac1{8\pi^2} \int \dd k_1 \dd k_2 \frac{\frac{\partial^2}{\partial^2\beta} h(\beta;k_1,k_2)}{g(\beta) + h(\beta;k_1,k_2)} \\
+ \frac1{8\pi^2} \int \dd k_1 \dd k_2 \biggl( \frac{g'(\beta) + \frac\partial{\partial\beta} h(\beta;k_1,k_2)}{g(\beta) + h(\beta;k_1,k_2)} \biggr)^2.
\end{multline}
For the first term we use the bounds $g(\beta) < g''(\beta_{\rm c}) (\beta-\beta_{\rm c})^2$ and $h(\beta;k_1,k_2) < {\rm const} (k_1^2 + k_2^2)$; using polar coordinates and \eqref{log divergence}, this term diverges as $\log|\beta-\beta_c|$ when $\beta \to \beta_{\rm c}$. The second term is easily seen to be bounded uniformly in $\beta \to \beta_{\rm c}$ using the second condition of the proposition and $|\frac{\partial^2}{\partial^2\beta}  h(\beta;k_1,k_2)| < {\rm const} (k_1^2 + k_2^2)$. For the third term we use \eqref{bound for ratio}, and $|\frac{\partial}{\partial\beta}  h(\beta;k_1,k_2)| < {\rm const} (k_1^2 + k_2^2)$. Using polar coordinates, and neglecting constants, we get an upper bound of the form
\bm
\int_0^1 \biggl( \frac{|\beta-\beta_{\rm c}| + r^2}{(\beta-\beta_{\rm c})^2 + r^2} \biggr)^2 r \dd r \\
\leq 4(\beta-\beta_{\rm c})^2 \int_0^{|\beta-\beta_{\rm c}|^{1/2}} \frac{r \dd r}{((\beta-\beta_{\rm c})^2 + r^2)^2} + \int_{|\beta-\beta_{\rm c}|^{1/2}}^1 \biggl( \frac{|\beta-\beta_{\rm c}| + r^2}{(\beta-\beta_{\rm c})^2 + r^2} \biggr)^2 r \dd r.
\end{multline}
The first integral is easily seen to behave as $|\beta-\beta_{\rm c}|^{-1}$ and it is controlled by the prefactor. The integrand of the second integral is a decreasing function of $r$; we get an upper bound by replacing $r$ with $|\beta-\beta_{\rm c}|^{1/2}$ which shows that it is bounded.

We have now verified that the only divergent term in \eqref{second derivative of f} is the first one, and the divergence is logarithmic indeed.
\end{proof}

\subsection{Case $J_1=J_2$}
\label{sec special case}

We consider the special case where two coupling constants are identical. By using symmetries (spin flips along alternate rows or columns) we can assume without loss of generality that $J_1=J_2 \geq 0$. Further, by rescaling $\beta$, we can take $J_1=J_2=1$.

\begin{theorem}
\label{thm special case}
Let $J_1=J_2=1$.
\begin{itemize}
\item[(a)] If $J_3 > -1$, there is a unique $\beta_{\rm c}$ such that $\sff(\beta)$ is analytic in $\bbR_+ \setminus \{ \beta_{\rm c} \}$ and $\sff''(\beta)$ has a logarithmic divergence at $\beta_{\rm c}$.
\item[(b)] If $J_3 \leq -1$, $\sff(\beta)$ is analytic in $\bbR_+$.
\end{itemize}
\end{theorem}

The theorem is illustrated with the phase diagram of Fig.\ \ref{fig phd}.

\begin{figure}[htb]\center
\includegraphics[width=9cm]{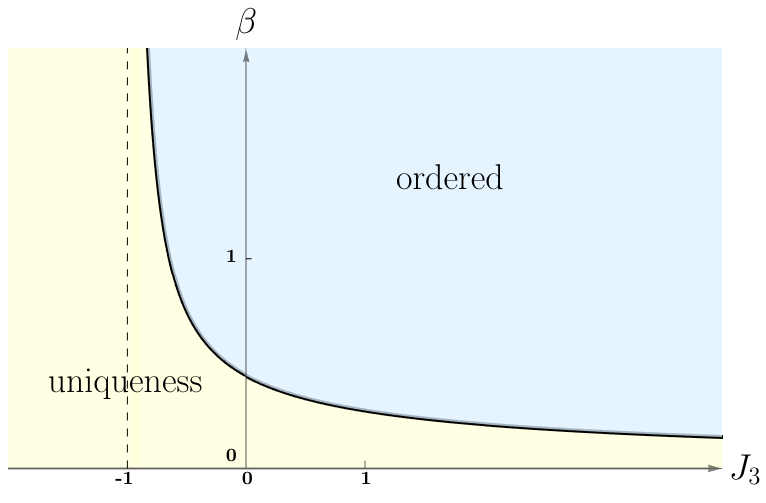}
\caption{Phase diagram with $J_1=J_2=1$. The free energy is proved to lack analyticity at the line that separates the ``ordered" and ``uniqueness" phases. The separation line is the inverse critical temperature $\beta_{\rm c} = \beta_{\rm c}(J_3)$; it is solution of the equation $\tanh\beta_{\rm c} = j^{-1}(J_3)$ with $j$ defined in Eq.\ \eqref{eq J_3}. For $J_3 \geq 0$, the article \cite{CDC} proves the existence of a unique infinite-volume Gibbs state for $\beta < \beta_{\rm c}$, and of several distinct Gibbs states for $\beta > \beta_{\rm c}$. For $J_3<0$, uniqueness is only proved for small $\beta$, and the existence of multiple Gibbs states is only proved for large $\beta$ (using the Pirogov-Sinai theory, see e.g.\ \cite{FV}).}
\label{fig phd}
\end{figure}

It helps to bring in the Cimasoni--Duminil-Copin--Li formula for the critical density, that was established for two-periodic planar lattices with nonnegative coupling constants \cite{Li, CDC}. Its general formulation involves sums over even graphs in the periodised cell that generates the lattice (see \cite[Theorem 1.1]{CDC}). In the present situation, this equation is
\be
\label{CDCL}
a(\beta) = 0
\ee
where the function $a(\beta)$ is defined as
\be
a(\beta) = 1 + \tanh^2\beta \tanh(\beta J_3) - 2\tanh\beta - \tanh(\beta J_3) - \tanh^2 \beta - 2\tanh\beta \tanh(\beta J_3).
\ee
In order to make the connection to the free energy \eqref{free energy 2}, we remark that
\be
g(\beta) = \cosh^2(2\beta) \cosh(2\beta J_3) a^2(\beta).
\ee
Then $g(\beta)$ vanishes precisely when $a(\beta)$ does; indeed, $h(\beta;k_1,k_2)$ is nonnegative when the coupling constants are nonnegative, and its minimum is 0. Proposition \ref{prop singularity} applies, establishing the singularity of the second derivative of the free energy.

We check in the proof below that the Cimasoni--Duminil-Copin--Li formula \eqref{CDCL} holds whenever $J_1 = J_2 \geq 0$, and $J_3 \in \bbR$ is allowed to be negative. One can also check that it {\it does not hold} if $J_1=J_2$ change signs; indeed, the free energy is the same due to symmetries (spin flips on a sublattice), but Eq.\ \eqref{CDCL} is different and has different solutions.

In addition to the non-analyticity of the free energy, Cimasoni and Duminil-Copin prove that the phase transition involves a change of the number of infinite-volume Gibbs states: There is just one for $\beta \leq \beta_{\rm c}$ and more than one for $\beta > \beta_{\rm c}$. The proof relies on the GKS and FKG correlation inequalities, which hold for nonnegative coupling constants only. It would be interesting to extend this to the case of coupling constants with arbitrary signs.

\begin{proof}[Proof of Theorem \ref{thm special case}]
When $J_3 \geq 0$ the theorem is a special case of \cite{CDC}, so we assume now that $J_3 \leq 0$ (even though the proof applies to positive $J_3$ with minor changes). We check that there exists a unique $\beta_{\rm c}$ that satisfies the conditions of Proposition \ref{prop singularity}.

We first check that $g(\beta) + h(\beta;k_1,k_2)$ can reach 0 only when $k_1=k_2=0$. Let $\alpha = -\frac{\sinh(2\beta J_3)}{\sinh(2\beta)}$. Using trigonometric identities, we have
\be
h(\beta;k_1,k_2) = 2\sinh(2\beta) + 2\sinh(2\beta J_3) + 2\sinh(2\beta) \bigl[ \alpha \cos^2(\tfrac{k_1+k_2}2) - \cos(\tfrac{k_1+k_2}2) \cos(\tfrac{k_1-k_2}2) \bigr].
\ee
We can minimise separately on the variables $\tfrac{k_1+k_2}2$ and $\tfrac{k_1-k_2}2$. There exists a minimiser satisfying  $\cos(\tfrac{k_1+k_2}2) \geq 0$ and $\cos(\tfrac{k_1-k_2}2)=1$. The minimum is then easily found, namely
\be
\min_{k_1,k_2} h(\beta;k_1,k_2) = \begin{cases} 0 & \text{if } \alpha \leq \frac12, \\ 2\sinh(2\beta) (1-\alpha - \frac1{4\alpha}) & \text{if } \alpha \geq \frac12. \end{cases}
\ee
The first case corresponds to $k_1=k_2=0$.  Suppose that $\alpha \geq \frac12$ and that $g(\beta) + \min h(\beta;k_1,k_2) = 0$. This is equivalent to
\be
(1 + \sinh^2(2\beta)) \sqrt{1+\alpha^2 \sinh^2(2\beta)} - (1+\sinh^2(2\beta)) \alpha \sinh(2\beta) - \frac{\sinh^2(2\beta)}{4\alpha} = 0.
\ee
The solution is $\alpha = \frac{\sinh(2\beta)}{2 \sqrt{1+\sinh^2(2\beta)}} < \frac12$; this contradicts the assumption that $\alpha \geq \frac12$.
This proves that, when $J_1=J_2$ and with arbitrary $J_3 \in \bbR$, the condition for $\beta_{\rm c}$ is $g(\beta_{\rm c}) = 0$, which is equivalent to the Cimasoni--Duminil-Copin--Li equation $a(\beta_{\rm c})=0$.

Instead of looking for $\beta_{\rm c}$ as function of $J_3$, it is more convenient to look for $J_3$ as function of $t=\tanh\beta$. The equation is then
\be
\label{eq J_3}
J_3 = \frac{\artanh \tfrac{1-2t-t^2}{1+2t-t^2}}{\artanh t} \equiv j(t).
\ee
The derivative of the function $j(t)$ is
\be
j'(t) = -\frac{(1+t^2) \artanh t + 2t \; \artanh \tfrac{1-2t-t^2}{1+2t-t^2} }{2t (1-t^2) \artanh^2 t}.
\ee
It is not hard to check that $\tfrac{1-2t-t^2}{1+2t-t^2} \geq -t$; it follows that the numerator above is positive so that $j'(t) < 0$. Further, $j(t)$ goes to $+\infty$ as $t\to0+$, and goes to $-1$ as $t\to1-$. Then $j^{-1}$ exists as a function $(-1,\infty) \to \bbR_+$; it follows that Eq.\ \eqref{eq J_3} has a unique solution when $J_3 > -1$ and no solutions otherwise. We also see that $\beta_{\rm c} \to 0$ as $J_3 \to \infty$, and $\beta_{\rm c} \to \infty$ as $J_3 \to -1$.

Finally we check that $g''(\beta_{\rm c}) > 0$. It is enough to check that $a'(\beta_{\rm c}) \neq 0$. We have
\be
a'(\beta) = -2(1-t^2) \bigl[ 1 + t  + (1-t) \tanh(\beta J_3) \bigr] - J_3 (1+2t-t^2) \bigl[ 1 - \tanh^2(\beta J_3) \bigr].
\ee
At $\beta = \beta_{\rm c}$ we have $\tanh(\beta_{\rm c} J_3) = \frac{1-2t-t^2}{1+2t-t^2}$, where $t = \tanh\beta_{\rm c}$. It is then possible to write $a'(\beta_{\rm c})$ as
\be
a'(\beta_{\rm c}) = -\frac{4(1-t^2) (1+t^2 + 2t J_3)}{1+2t-t^2}.
\ee
This is clearly not 0 since $t<1$ and $J_3 > -1$.

The condition on $h$ in Proposition \ref{prop singularity} clearly holds.
\end{proof}

\section{The Kac-Ward identity}
\label{sec Kac-Ward}

We rely on the extension of the Kac-Ward identity to ``faithful projections" of non-planar graphs. It was proposed by Cimasoni \cite{Cim} and used in \cite{KLM, AW}. In order to accommodate negative weights we need two faithful projections for $\bbT_{L,M}$ with edges between nearest-neighbours. The graphs are $G_1$ and $G_2$ and they are illustrated in Fig.\ \ref{fig graph}. Here is a full description of the left graph:
\begin{itemize}
\item The vertices are $(i,j)$ with $1 \leq i \leq L$ and $1 \leq j \leq M$.
\item There are edges represented by straight lines between $(i,j)$ and $(i+1,j)$ for $1 \leq i \leq L-1$, $1 \leq j \leq M$; between $(i,j)$ and $(i,j+1)$ for $1 \leq i \leq L$, $1 \leq j \leq M-1$; and between $(i,j)$ and $(i+1,j+1)$ for $1 \leq i \leq L-1$, $1 \leq j \leq M-1$.
\item There are edges represented by ``handles" (continuous curves with winding number $-1$) between $(L,j)$ and $(1,j)$ for $1 \leq j \leq M$; between $(L,j)$ and $(1,j+1)$ for $1 \leq j \leq M-1$; between $(i,M)$ and $(i,1)$ for $1 \leq i \leq L$; between $(i,M)$ and $(i+1,1)$ for $1 \leq i \leq L-1$.
\item And there is a self-crossing handle between $(L,M)$ and $(1,1)$ whose winding number is $-2$.
\item The handles are drawn so that handles starting at $(i,M)$ only cross the handles starting at $(L,j)$ (and they cross them exactly once); the self-crossing handle belongs to both groups.
\end{itemize} 
The second graph is similar, except that the oblique handle no longer self-crosses but the other horizontal handles all self-cross.

\begin{figure}[htb]\center
\includegraphics[width=65mm]{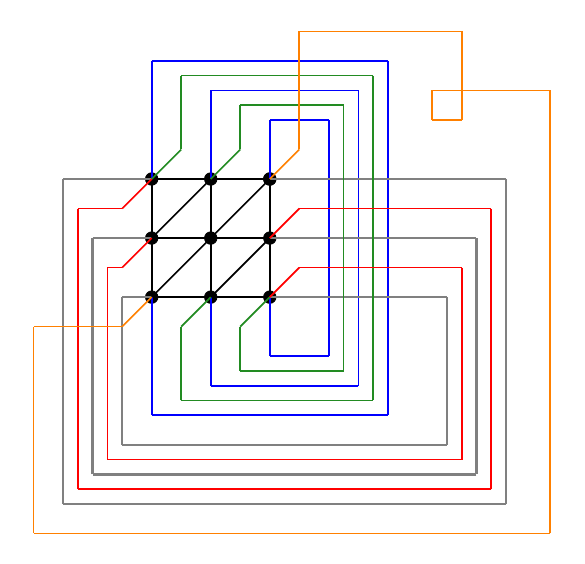}\hspace{12mm}\includegraphics[width=65mm]{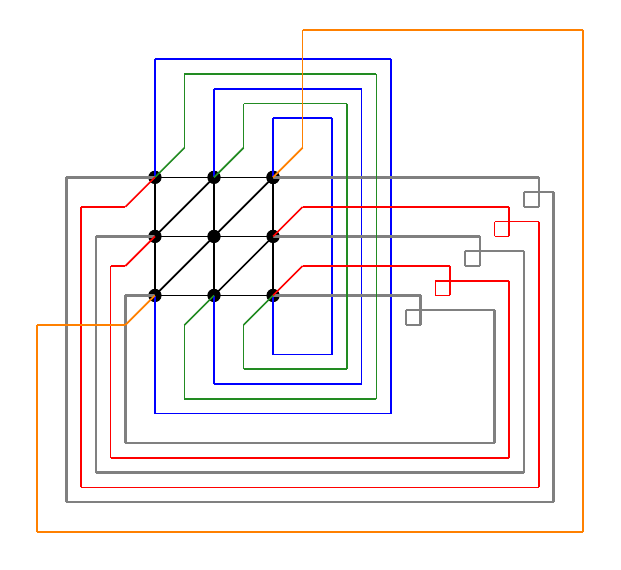}
\caption{Two faithful projections of the graph $(\bbT_{3,3},\caE_{3,3})$. The handles cross at non-vertex locations; some handles cross themselves. The matrix $K^{\scriptscriptstyle (1)}$ is defined on the left graph; the matrix $K^{\scriptscriptstyle (2)}$ is defined on the right graph.}
\label{fig graph}
\end{figure}

The Kac-Ward identity involves matrices indexed by directed edges. We denote $\vec\caE_{L,M}$ the edges of $\caE_{L,M}$ with direction. The coupling constants defined in Eq.\ \eqref{def couplings} can be extended to directed edges by assigning the same value $J_e$ to both directions of the same edge; then we let $W$ to be the diagonal matrix whose element $W_{e,e}$ is equal to $\tanh J_e$. We now introduce the Kac-Ward matrix $K^{\scriptscriptstyle (1)}$ by
\be\label{K1}
K_{e,e'}^{\scriptscriptstyle (1)} = 1_{e \, \triangleright \, e'} \, \e{\frac\ii2 \measuredangle_1(e,e') + \frac\ii2 \measuredangle_1(e)}, \quad e,e' \in \caE_{L,M}.
\ee
Here $e \, \triangleright \, e'$ means that the endpoint of $e$ is equal to the starting point of $e'$ and also that $e'$ is not equal to the reverse of $e'$ (the matrix is not ``backtracking"). $\measuredangle_1(e,e'):\mathcal{E}_{L,M} \rightarrow(-\pi,\pi]$ is the angle between the end of $e$ and the start of $e'$ on the the faithful projection $G_1$; $\measuredangle_1(e):\mathcal{E}_{L,M} \rightarrow \mathbb{R}$ is the integrated angle along the planar curve that represents the edge $e$.

Following \cite{AW} we define an average over even subgraphs: If $f$ is a function on graphs, let
\be
\langle f \rangle_{L,M} = \frac1{\widetilde Z_{L,M}} \sum_{\Gamma \subset \caE_{L,M}: \partial \Gamma = \emptyset} f(\Gamma) w(\Gamma)
\ee
where the normalisation is $\widetilde Z_{L,M} = \sum_{\Gamma \in \caE_{L,M}: \partial \Gamma = \emptyset} w(\Gamma)$. The definition of the weight is $w(\Gamma) = \prod_{e \in \Gamma} \tanh J_e$. The boundary $\partial\Gamma$ of a graph is the set of vertices whose incidence number is odd; the sum in the right hand side is over even subgraphs. 
Notice that $\widetilde Z_{L,M}$ is always positive as can be seen from its relation to the Ising partition function, see \eqref{high T exp} below.

With these definition, we have the remarkable Kac-Ward identity \cite[Theorem 5.1]{AW}:
\be
\label{Kac-Ward}
\sqrt{\det(1 - K^{\scriptscriptstyle (1)}W)} = \widetilde Z_{L,M}  \big\langle (-1)^{n^{\scriptscriptstyle (1)}_0(\Gamma)} \big\rangle_{L,M}.
\ee
Here $n^{\scriptscriptstyle (1)}_0(\Gamma)$ is the total number of crossings between all edges of $\Gamma$ when the graph is projected on $G_1$.   

It is worth noting that the right side of \eqref{Kac-Ward} is a multinomial in $(W_{e,e})$, something that is not apparent in the left side --- there are remarkable cancellations indeed. This allows \cite{AW} to prove the identity for small $(W_{e,e})$; the extension to larger values is automatic. The determinant cannot be negative and the sign of the square root cannot change.

We define the matrix $K^{\scriptscriptstyle (2)}$ as in \eqref{K1} but  $\measuredangle_2(e,e')$ and $ \measuredangle_2(e)$ are the  corresponding angles on the the faithfull projection $G_2$. Analogously, we define $n^{\scriptscriptstyle (2)}_0(\Gamma)$ for this projection. 

The connection with the Ising model is through the high-temperature expansion, see e.g.\ \cite[Section 3.7.3]{FV}. The partition function \eqref{def part fct} is equal to
\be
\label{high T exp}
Z_{L,M}(J_1,J_2,J_3) = 2^{LM} \biggl( \prod_{e \in \caE_{L,M}} \cosh J_e \biggr) \sum_{\Gamma \subset \caE_{L,M}: \partial \Gamma = \emptyset} w(\Gamma) = 2^{LM} \biggl( \prod_{e \in \caE_{L,M}} \cosh J_e \biggr) \widetilde Z_{L,M}.
\ee

The strategy of Aizenman and Warzel \cite{AW} is to prove that $\langle (-1)^{n^{\scriptscriptstyle (1)}_0(\Gamma)} \rangle_{L,M} \to 1$ as $L,M \to \infty$. This can be done when the coupling constants are positive, and small enough so the temperature is higher than the 2D critical temperature. (Then duality is used to get the formula for low temperatures.) The presence of negative coupling constants necessitates a different approach. We first show in Lemma \ref{lem Georgios} that a combination of the two faithful projections gives the partition function, up to a correction. We then show in Lemma \ref{lem Perron Frobenius} that this correction vanishes in the limit $L\to\infty$, for fixed $M$. Denote by $n_{\rm h}(\Gamma)$ the number of horizontal handles of the subgraph $\Gamma$, that is, the number of handles in $\Gamma$ that connect sites of the form $(L,i)$ with sites $(1,j)$. Note that the total number of horizontal handles of $\mathcal{E}_{L,M}$ is $2M$.

\begin{lemma}
\label{lem Georgios}
We have
\[
\sqrt{\det(1-K^{\scriptscriptstyle (1)}W)} + \sqrt{\det(1-K^{\scriptscriptstyle (2)}W)} = 2 \widetilde Z_{L,M} \Bigl( 1 - \big\langle 1_{n_{\rm h}(\Gamma) \, \rm{odd}} \big\rangle_{L,M} \Bigr).
\]
\end{lemma}

\begin{proof}
From Eq.\ \eqref{Kac-Ward}, we have
\begin{equation}\label{2dets}
    \sqrt{\det(1-K^{\scriptscriptstyle (1)}W)} + \sqrt{\det(1-K^{\scriptscriptstyle (2)}W)}=\widetilde Z_{L,M}  \big\langle (-1)^{n^{\scriptscriptstyle (1)}_0(\Gamma)}+(-1)^{n^{\scriptscriptstyle (2)}_0(\Gamma)}  \big\rangle_{L,M}
\end{equation}

Let $n_{\rm v}(\Gamma)$ be the number of handles in $\Gamma$ that connect sites of the form $(i,M)$ with sites $(j,1)$ (excluding the handle between $(L,M)$ and $(1,1)$) and let $n_{\rm hv}(\Gamma) = 0,1$ be the indicator on whether the handle from $(L,M)$ and $(1,1)$ is present (notice the asymmetric definition of $n_{\rm v}$ and $n_{\rm h}$, as the oblique handle is included in $n_{\rm h}$ but not in $n_{\rm v}$). We have
\be
\begin{split}
 \quad & 1_{n^{\scriptscriptstyle (1)}_0(\Gamma) \, \rm {odd}} = 1_{n_{\rm h}(\Gamma) \, \rm{odd}} \; \bigl( 1_{n_{\rm hv}(\Gamma)=0} \; 1_{n_{\rm v}(\Gamma) \, {\rm odd}} + 1_{n_{\rm hv}(\Gamma)=1} \; 1_{n_{\rm v}(\Gamma) \, {\rm even}} \bigr);\\
 \quad & 1_{n^{\scriptscriptstyle (2)}_0(\Gamma) \, \rm {odd}} = 1_{n_{\rm h}(\Gamma) \, \rm{odd}} \; \bigl( 1_{n_{\rm hv}(\Gamma)=0} \; 1_{n_{\rm v}(\Gamma) \, {\rm even}} + 1_{n_{\rm hv}(\Gamma)=1} \; 1_{n_{\rm v}(\Gamma) \, {\rm odd}} \bigr).
\end{split}
\ee
It follows that
\be
1_{n^{\scriptscriptstyle (1)}_0(\Gamma) \, \rm {odd}} + 1_{n^{\scriptscriptstyle (2)}_0(\Gamma) \, \rm {odd}} = 1_{n_{\rm h}(\Gamma) \, \rm{odd}}.
\ee
By combining the above relation with \eqref{2dets}, using $(-1)^{n_0\scr{i}(\Gamma)} = 1 - 2 \cdot 1_{n_0\scr{i} (\Gamma) \, {\rm odd}}$, the lemma follows.
\end{proof}

\begin{lemma}
\label{lem Perron Frobenius}
For any $J_1, J_2, J_3 \in \bbR$, for any $M \in \bbN$, we have
\[
\lim_{L \to \infty} \big\langle 1_{n_{\rm h}(\Gamma) \, \rm{odd}} \big\rangle_{L,M} = 0.
\]
\end{lemma}

\begin{proof}
We condition on the horizontal handles (including possibly the self-crossing ones). We denote by $\frah$ the set of handles that connect sites in the leftmost and rightmost columns:
\be
\frah = \bigl\{ \{(1,j_1),(L,j_1')\}, \dots, \{(1,j_k),(L,j_k')\} \bigr\}.
\ee
Then we define the support $\supp_1 \frah$, resp.\  $\supp_L \frah$, to be the set of vertices of the form $(1,j_i)$, resp.\  $(L,j_i')$, that appear an odd number of times in $\frah$. Let $\supp\frah = \supp_1 \frah \cup \supp_L \frah$. We let $\tilde\caE_{L,M}$ be the set of edges of the cylinder (not the torus) $\{1,\dots,L\} \times \bbT_M$. With $1_\frah = 1_\frah(\Gamma)$ the indicator function that the random graph $\Gamma$ has set of handles $\frah$, we have
\be
\begin{split}
\big\langle 1_{n_{\rm h}(\Gamma) \, \rm{odd}} \big\rangle_{L,M} &= \sum_{|\frah| \; {\rm odd}} \langle 1_\frah \rangle_{L,M} \\
&= \sum_{|\frah| \; {\rm odd}} \biggl( \prod_{i=1}^k \tanh J_{(1,j_i),(L,j_i')} \biggr) \frac1{\widetilde Z_{L,M}} \sum_{\Gamma \subset \tilde\caE_{L,M} : \partial\Gamma = \supp\frah} w(\Gamma).
\end{split}
\ee
We now consider an Ising model on the cylinder $\{1,\dots,L\} \times \bbT_M$. We have
\be
 \frac1{\widetilde Z_{L,M}^{\rm cyl}} \sum_{\Gamma \subset \tilde\caE_{L,M} : \partial\Gamma = \supp\frah} w(\Gamma) = \Big\langle \prod_{x \in \supp\frah} \sigma_x \Big\rangle_{L,M}^{\rm cyl}.
 \ee
 Notice that the partition function $\widetilde Z_{L,M}^{\rm cyl}$ is almost equal to $\widetilde Z_{L,M}$; either ratio is less than $\e{2M(|J_1|+|J_3|)}$.
Next we introduce the transfer matrix $T_{\eta,\eta'}$ between column configurations $\eta,\eta' \in \{-1,+1\}^M$:
\be
T_{\eta,\eta'} = \exp\biggl\{ \sum_{i=1}^M \Bigl( J_1 \eta_i \eta_i' + J_2 \eta_i \eta_{i+1} + J_3 \eta_i \eta_{i+1}' \Bigr) \biggr\}.
\ee
Here we defined $\eta_{M+1} \equiv \eta_1$. The transfer matrix allows to write the Ising correlations above as
\be
\label{transfer matrix expression}
\Big\langle \prod_{x \in \supp\frah} \sigma_x \Big\rangle_{L,M}^{\rm cyl} = \frac1{T^L} \sum_{\eta,\eta'} \langle \eta | T^{L-1} | \eta' \rangle \Bigl( \prod_{x \in \supp_1 \frah} \eta_x \Bigr) \Bigl( \prod_{y \in \supp_L \frah} \eta_y' \Bigr) \e{J_2 \sum_{i=1}^M \eta_i' \eta_{i+1}'}.
\ee
The matrix elements of $T$ are positive; by the Perron-Frobenius theorem there exist vectors $|v\rangle, |w\rangle$ such that
\be
\lim_{L\to\infty} \frac{T^{L-1}}{\Tr T^L} = \e{-\lambda_{\rm max}} |v\rangle \langle w|.
\ee
Here $\lambda_{\rm max}$ is the largest eigenvalue of $T$ (it depends on $M$). The vectors $|v\rangle, |w\rangle$ can be decomposed in the basis $\{ |\eta\rangle \}$ of column configurations and their coefficients have the spin-flip symmetry. Taking the limit $L\to\infty$ in \eqref{transfer matrix expression} one gets 0. Indeed, the sum over $\eta$ is
\be
\sum_\eta \Bigl( \prod_{x \in \supp_1 \frah} \eta_x \Bigr) \langle \eta | v \rangle
\ee
which is zero since $\supp_1 \frah$ contains an odd number of vertices; the sum over $\eta'$ also gives zero.
\end{proof}

Next we seek to calculate the determinants of $1-K^{\scriptscriptstyle (1)}W$ and $1-K^{\scriptscriptstyle (2)}W$. For this we first make the matrices translation-invariant so we can use the Fourier transform. Let us define $\widetilde K^{\scriptscriptstyle (i)}$ $i=1,2$ to be as $K^{\scriptscriptstyle (i)}$ $i=1,2$ but omitting the respective integrated angle of the handles:
\be
\widetilde K_{e,e'}^{\scriptscriptstyle (i)} = 1_{e \, \triangleright e'} \e{\frac\ii2 \measuredangle_i(e,e')} \hspace{4mm} i=1,2.
\ee
Actually $\widetilde K_{e,e'}^{\scriptscriptstyle (1)}=\widetilde K_{e,e'}^{\scriptscriptstyle (2)}$ and we shall write $\widetilde K_{e,e'}$ for either $\widetilde K_{e,e'}^{\scriptscriptstyle (1)}$ or $\widetilde K_{e,e'}^{\scriptscriptstyle (2)}$.
Then we define modified diagonal matrices $\widetilde W_{e,e}^{\scriptscriptstyle (1)}$ and $\widetilde W_{e,e}^{\scriptscriptstyle (2)}$; matrix elements now depend on the direction of $e$:
\be
\widetilde W^{\scriptscriptstyle (1)}_{e,e} = \begin{cases} W_{e,e} \e{\ii \pi / L} & \text{if } e = \rightarrow, \\ W_{e,e} \e{-\ii \pi / L} & \text{if } e = \leftarrow, \\ W_{e,e} \e{\ii \pi / M} & \text{if } e = \uparrow, \\ W_{e,e} \e{-\ii \pi / M} & \text{if } e = \downarrow, \\ W_{e,e} \e{\ii \pi (\frac1L+\frac1M)} & \text{if } e = \nearrow, \\ W_{e,e} \e{-\ii \pi (\frac1L+\frac1M)} & \text{if } e = \swarrow.
\end{cases}
\qquad
\widetilde W^{\scriptscriptstyle (2)}_{e,e} = \begin{cases} W_{e,e} & \text{if $e = \rightarrow$ or $\leftarrow$}, \\ W_{e,e} \e{\ii \pi / M} & \text{if $e = \uparrow$ or $\nearrow$}, \\ W_{e,e} \e{-\ii \pi / M} & \text{if $e = \downarrow$ or $\swarrow$}.
\end{cases}
\ee

\begin{lemma}\label{transl.inv.}
We have
\[
\det(1-K^{\scriptscriptstyle (1)}W) = \det(1-\widetilde K \widetilde W^{\scriptscriptstyle (1)}), \qquad \det(1-K^{\scriptscriptstyle (2)}W) = \det(1-\widetilde K \widetilde W^{\scriptscriptstyle (2)}).
\]
\end{lemma}

\begin{proof}
One can expand the determinants as products of directed loops as in \cite[Theorem 3.2]{AW}. Let $\gamma = (e_1,\dots,e_k)$ be a directed loop with $\ell$ handles (the self-crossing handle between $(L.M)$ and $(1,1)$ is counted twice). We have
\be
\prod_{i=1}^k K^{\scriptscriptstyle (1)}_{e_i,e_{i+1}} = (-1)^\ell \, \prod_{i=1}^k \widetilde K_{e_i,e_{i+1}}, \qquad \prod_{i=1}^k \widetilde W^{\scriptscriptstyle (1)}_{e_i,e_i} = (-1)^\ell \, \prod_{i=1}^k W_{e_i,e_i}.
\ee
Then each loop gives the same contribution in $\det(1-K^{\scriptscriptstyle (1)}W)$ and in $\det(1-\widetilde K \widetilde W^{\scriptscriptstyle (1)})$. The argument for $\det(1-K^{\scriptscriptstyle (2)}W)$ is the same, counting only vertical and oblique handles between sites $(i,M)$ and $(j,1)$, $1\leq i, j, \leq L$.
\end{proof}

\begin{lemma}\label{dets}
Let $\bbT_{L}^*=\frac{2\pi}L \bbT_L $ and recall that $\widetilde{\bbT}_L = \bbT_{L}^* + \frac\pi{L}$.
\begin{itemize}
\item[(a)] With $k_3 = k_1 + k_2$, we have
\begin{multline*}
\det(1-\widetilde K \widetilde W^{\scriptscriptstyle (1)})= \prod_{k_1 \in \widetilde{\bbT}_L} \prod_{k_2 \in \widetilde{\bbT}_M} \bigg[ \prod_{i=1}^3 \big( 1+\tanh^2 J_i \big)+8\prod_{i=1}^3\tanh J_i \\
    -2\sum_{i=1}^{3}\tanh J_i \big(1-\tanh^2 J_{i+1} \big) \big(1-\tanh^2 J_{i+2} \big) \cos k_i \bigg].
\end{multline*}
\item[(b)] Again with $k_3 = k_1 + k_2$, we have
\begin{multline*}
    \det(1-\widetilde K \widetilde W^{\scriptscriptstyle (2)}) = \prod_{k_1 \in \bbT_{L}^*} \prod_{k_2 \in \widetilde{\bbT}_M} \bigg[ \prod_{i=1}^3 \big(1+\tanh^2 J_i \big)+8\prod_{i=1}^3\tanh J_i \\
    - 2\sum_{i=1}^{3} \tanh J_i \big(1-\tanh^2 J_{i+1} \big)\big(1-\tanh^2 J_{i+2} \big) \cos k_i \bigg].
\end{multline*}
\end{itemize}
\end{lemma}

\begin{proof}
For (a) we label the set of directed edges as $(x,\alpha)$ where $x \in \bbT_{L,M}$ and $\alpha \in A$, with
\be
A = \bigl\{ \rightarrow, \leftarrow, \uparrow, \downarrow, \nearrow, \swarrow \bigr\}.
\ee
The Fourier coefficients are $(k,\alpha)$ with $k \in \bbT_{L,M}^* = \bbT_{L}^* \times \bbT_{M}^*$. The Fourier transform is represented by the unitary matrix $U$:
\be
U_{(k,\alpha),(x,\beta)} = \frac1{\sqrt{LM}} \e{-\ii k x} \delta_{\alpha,\beta},
\ee
for $x \in \bbT_{L,M}$, $k \in \bbT_{L,M}^*$, and $\alpha, \beta \in A$. Since $\widetilde W_{e,e}$ depends only on $\alpha \in A$, we have
\be
(U \widetilde W^{\scriptscriptstyle (1)} U^{-1)})_{(k,\alpha),(k',\beta)} = \widetilde W^{\scriptscriptstyle (1)}_\alpha \delta_{k,k'} \delta_{\alpha,\beta}.
\ee
Further, straightforward Fourier calculations give
\be
(U \widetilde K U^{-1})_{(k,\alpha),(k',\beta)} = \delta_{k,k'} \sum_{x \in \bbT_{L,M}} \e{\ii kx} \widetilde K_{(0,\alpha),(x,\beta)} \equiv \delta_{k,k'} \widehat K_{\alpha,\beta}(k),
\ee
with the matrix $\widehat K(k)$ given by
\be
\widehat K(k) = \left( \begin{matrix} \e{\ii k_1} \!\!\! & 0 & 0 & 0 & 0 & 0 \\ 0 & \!\!\! \e{-\ii k_1} \!\!\! & 0 & 0 & 0 & 0 \\ 0 & 0 & \!\!\! \e{\ii k_2} \!\!\! & 0 & 0 & 0 \\ 0 & 0 & 0 & \!\!\! \e{-\ii k_2} \!\!\! & 0 & 0 \\ 0 & 0 & 0 & 0 & \!\!\! \e{\ii (k_1+k_2)} \!\!\!\! & 0 \\ 0 & 0 & 0 & 0 & 0 & \!\!\!\! \e{-\ii (k_1+k_2)} \end{matrix} \right) \;
\left( \begin{matrix} 1 & 0 & \!\! \e{\ii \frac\pi4} \!\! & \!\! \e{-\ii \frac\pi4} \!\! & \!\! \e{\ii \frac\pi8} \!\! & \!\! \e{-\ii \frac{3\pi}8} \\ 0 & 1 & \!\! \e{-\ii \frac\pi4} \!\! & \!\! \e{\ii \frac\pi4} \!\! & \!\! \e{-\ii \frac{3\pi}8} \!\! & \!\! \e{\ii \frac\pi8} \\ \e{-\ii \frac\pi4} \!\! & \!\! \e{\ii \frac\pi4} \!\! & 1 & 0 & \!\! \e{-\ii \frac\pi8} \!\! & \!\! \e{\ii \frac{3\pi}8} \\ \e{\ii \frac\pi4} \!\! & \!\! \e{-\ii \frac\pi4} \!\! & 0 & 1 & \!\! \e{\ii \frac{3\pi}8} \!\! & \!\! \e{-\ii \frac\pi8} \\ \e{-\ii \frac\pi8} \!\! & \!\! \e{\ii \frac{3\pi}8} \!\! & \!\! \e{\ii \frac\pi8} \!\! & \!\! \e{-\ii \frac{3\pi}8} \!\! & 1 & 0 \\ \e{\ii \frac{3\pi}8} \!\! & \!\! \e{-\ii \frac\pi8} \!\! & \!\! \e{-\ii \frac{3\pi}8} \!\! & \!\! \e{\ii \frac\pi8} \!\! & 0 & 1 \end{matrix} \right).
\ee
Let us define
\be
\widehat W^{\scriptscriptstyle (1)}(k) := \left( \begin{matrix} t_1 \e{\ii k_1} \!\!\! & 0 & 0 & 0 & 0 & 0 \\ 0 & \!\!\! t_2 \e{-\ii k_1} \!\!\! & 0 & 0 & 0 & 0 \\ 0 & 0 &\!\!\! t_2 \e{\ii k_2} \!\!\! & 0 & 0 & 0 \\ 0 & 0 & 0 &\!\!\! t_2 \e{-\ii k_2} \!\!\! & 0 & 0 \\ 0 & 0 & 0 & 0 &\!\!\! t_3 \e{\ii (k_1+k_2)} \!\!\! & 0 \\ 0 & 0 & 0 & 0 & 0 &\!\!\! t_3 \e{-\ii (k_1+k_2)} \end{matrix} \right) \;
\ee
where $t_i=\tanh{J_i}$, $i=1,2,3$.
Then
\be
\begin{split}
\det{(1-\widetilde{K}\widetilde{W}^{\scriptscriptstyle (1)})} &= \det{(1-\widetilde{W}^{\scriptscriptstyle (1)}\widetilde{K})} = \prod_{k \in \bbT_{L,M}^*} \det \Bigl[ 1-\widehat W^{\scriptscriptstyle (1)}\bigl(k+(\tfrac{\pi}{L},\tfrac{\pi}{M})\bigr)\widehat K(0) \Bigr] \\
&= \prod_{k_1 \in \widetilde{\bbT}_L} \prod_{k_2 \in \widetilde{\bbT}_M} \det\Bigl[ 1 - \widehat{W}^{\scriptscriptstyle (1)}\bigl((k_1,k_2) \bigr)\widehat{K}(0)) \Bigr].
\end{split}
\ee
The first identity follows from a loop expansion, see \cite[Theorem 3.2]{AW}.
A calculation of the determinant by grouping the terms according to $k_1+k_2,k_1,k_2$ yields
\begin{multline}
\label{det}
\det\bigl[ 1 - \widehat{W}^{\scriptscriptstyle (1)} \bigl((k_1,k_2) \bigr)\widehat{K}(0)) \bigr]
= \prod_{i=1}^3 \big( 1 + \tanh^2 J_i \big)+8\prod_{i=1}^3 \tanh J_i \\
- 2\sum_{i=1}^{3} \tanh J_i \big(1-\tanh^2 J_{i+1} \big) \big( 1 - \tanh^2 J_{i+2} \big) \cos k_i
\end{multline}
where $k_3=k_1+k_2$. This gives (a).

The proof of (b) is similar.
\end{proof}

\begin{corollary}\hfill
\label{Cor}

\begin{itemize}
\item[(a)] The determinants are nonnegative, $\det(1-\widetilde K \widetilde W^{\scriptscriptstyle (1)}) \geq 0$ and  $\det(1-\widetilde K \widetilde W^{\scriptscriptstyle (2)}) \geq 0$.

\item[(b)] Taking the logarithms, dividing by $L$, we have as $L\to\infty$
\[
\begin{split}
\lim_{L\to\infty} \frac1L \log & \det(1 - \widetilde K \widetilde W^{\scriptscriptstyle (1)}) = \lim_{L\to\infty} \frac1L \log \det(1-\widetilde K \widetilde W^{\scriptscriptstyle (2)}) \\
&= \int_{[-\pi,\pi]} \dd k_1 \sum_{k_2 \in \widetilde{\bbT}_M} \log \biggl[ \prod_{i=1}^3\big(1+\tanh^2{ J_i }\big) + 8\prod_{i=1}^3\tanh J_i \\
&\qquad -\sum_{i=1}^{3}2\tanh J_i \big(1-\tanh^2 J_{i+1} \big)\big(1-\tanh^2 J_{i+2} \big)\cos k_i \biggr].
\end{split}
\]
\end{itemize}
\end{corollary}

\begin{proof}
(a) By Eq.\ (\ref{Kac-Ward}) and Lemma \ref{transl.inv.}, we obtain that both square roots of the above determinants are real. (b) This is a consequence of Lemma \ref{dets}; taking the logarithm we obtain Riemann sums.
\end{proof}

\begin{proof}[Proof of Theorem \ref{thm triangle Ising}]
(a) From  the high temperature expansion (\ref{high T exp}), we observe that the finite volume free energy with periodic boundary conditions satisfies 
\be
-f_{L,M}(J_1,J_2,J_3)= \log 2 + \log\Big[\prod_{i=1}^3\cosh J_i \Big]+\frac{1}{LM}\log\Big[\widetilde Z_{L,M} \Big].
\ee
Using Lemma \ref{lem Georgios}, Lemma \ref{lem Perron Frobenius} and Lemma \ref{transl.inv.}, we see that the free energy on the infinite cylinder is
\begin{multline}\label{inf. cyl}
-f_M(J_1,J_2,J_3)=\log 2 + \log\Big[\prod_{i=1}^3\cosh{J_i}\Big] \\
+\lim_{L \to \infty}\frac{1}{LM}\log\Biggl[\sqrt{\det(1-\widetilde K \widetilde W^{\scriptscriptstyle (1)})}+\sqrt{\det(1-\widetilde K \widetilde W^{\scriptscriptstyle (2)})}\Biggr]
\end{multline}

By Corollary~\ref{Cor} (a) we have
\be
\begin{split}
\log \sqrt{\det(1-\widetilde K \widetilde W\scr1)} &\leq \log \Biggl[ \sqrt{\det(1-\widetilde K \widetilde W\scr1)} +  \sqrt{\det(1-\widetilde K \widetilde W\scr2)} \Biggr] \\
&\leq \max_{i=1,2} \log \sqrt{\det(1-\widetilde K \widetilde W\scr{i})} + \log2.
\end{split}
\ee
Dividing by $L$, all terms above converge to the same limit as $L\to\infty$ by Corollary~\ref{Cor} (b).
We get
\begin{multline}\label{3.24}   
-f(J_1,J_2,J_3)=\log 2 + \log\Big[\prod_{i=1}^3\cosh{J_i}\Big] +\frac{1}{4\pi M}\int_{[0,2\pi]}dk_1\sum_{k_2 \in \widetilde{\bbT}_M}\log\bigg[\prod_{i=1}^3\big(1 + \tanh^2 J_i \big) \\
+8\prod_{i=1}^3 \tanh J_i
-\sum_{i=1}^{3} 2 \tanh J_i \big(1-\tanh^2 J_{i+1} \big)\big(1-\tanh^2 J_{i+2} \big) \cos k_i \bigg].
\end{multline}
In order to get the expression of Theorem \ref{thm triangle Ising}, one should use the hyperbolic identities $1+\tanh^2 x = \frac{\cosh(2x)}{\cosh^2 x}$ and $\tanh x = \frac{\sinh 2x}{2\cosh^2 x}$, and extract a factor $\bigl( \prod_i \cosh J_i \bigr)^{-1}$.
\end{proof}

\section{The 1D quantum Ising model}  
\label{sec qu Ising}

One application of the cylinder formula of Theorem \ref{thm triangle Ising} (a) deals with the one-dimensional quantum Ising model. It is well-known that it can be mapped to a classical model in $1+1$ dimensions, the extra dimension being the continuous interval $[0,\beta]$ with periodic boundary conditions. A phase transition is only possible when both dimensions are infinite, which necessitates taking the limit of zero-temperature $\beta \to \infty$. The free energy of the quantum Ising model was first computed by Pfeuty \cite{Pfe} using the fermionic method of \cite{SML}. The results of this section are not new, but the Kac-Ward approach may have more appeal to some readers.

We consider the chain $\{1,\dots,L\}$ with periodic boundary conditions. The Hilbert space is $\caH_L = \otimes_{i=1}^L \bbC^2$.
Let $S\scr1$ and $S\scr3$ denote the spin operators on $\bbC^2$ whose matrices are
\be
S\scr1 = \tfrac12 \biggl( \begin{matrix} 0 & 1 \\ 1 & 0 \end{matrix} \biggr), \qquad S\scr3 = \tfrac12 \biggl( \begin{matrix} 1 & 0 \\ 0 & -1 \end{matrix} \biggr).
\ee
Then we denote $S_i\scr j$ the spin operators at site $i \in \bbZ$. With $h \in \bbR$ the magnetic field, the hamiltonian is
\be
H_L = -\sum_{i=1}^L S_i\scr3 S_{i+1}\scr3 - h \sum_{i=1}^L S_i\scr1.
\ee
Here the site $i = L+1$ is defined as $i=1$. The partition function is
\be
Z_L^{\rm qu}(\beta,h) = \Tr_{\caH_\Lambda} \e{-\beta H_L}.
\ee
The finite-volume free energy is
\be
f_L^{\rm qu}(\beta,h) = -\frac1{\beta L} \log Z_L^{\rm qu}(\beta,h).
\ee
Notice the division by $\beta$, which allows to get the ground state energy by taking the limit $\beta\to\infty$.

\begin{theorem}
\label{thm quantum Ising}
The infinite-volume free energy of the one-dimensional quantum Ising model is equal to
\[
f^{\rm qu}(\beta,h) = \lim_{L\to\infty} f_L^{\rm qu}(\beta,h) = -\tfrac1\beta \log 2 - \frac1{2\pi \beta} \int_{-\pi}^\pi \dd k \log \cosh \Bigl( \frac\beta4 \sqrt{1 + 4h^2 + 4h \cos k} \Bigr).
\]
\end{theorem}

We prove this theorem by invoking the well-known fact that the $d$-dimensional quantum Ising model is equivalent to a $(d+1)$-dimensional classical Ising model, the extra dimension being continuous; see Proposition \ref{prop qu-class}. We check in Proposition \ref{prop limit interchange} that the continuum limit can be taken {\it after} the infinite-volume limit. This allows to make direct use of Theorem \ref{thm triangle Ising}. The remaining step is to take the continuum limit and it is not entirely straightforward; the proof of Theorem \ref{thm quantum Ising} can be found at the end of this section.

\begin{proposition}
\label{prop qu-class}
Let us define coupling constants $J_1\scr n, J_2\scr n$ by
\[
J_1\scr n = \frac\beta{4n}, \qquad J_2\scr n = -\tfrac12 \log \frac{\beta h}{2n}.
\]
Then we have the identity
\[
Z_L^{\rm qu}(\beta,h) = \lim_{n\to\infty} Z_{L,n}^{\rm qu}(\beta,h) \\
\]
with
\[
Z_{L,n}^{\rm qu}(\beta,h) = \exp\Bigl\{ \tfrac12 Ln \log \tfrac{\beta h}{2n} \Bigr\} \; Z_{L,n}(J_1\scr n, J_2\scr n).
\]
Here $Z_{L,n}(J_1\scr n, J_2\scr n)$ is the partition function defined in Eq.\ \eqref{def part fct} with $J_3=0$.
\end{proposition}

\begin{proof}
By the Lie-Trotter formula,
\be
\label{Trotter expansion}
\begin{split}
\Tr \e{-\beta H_L} &= \lim_{n\to\infty} \Tr \biggl( \e{\frac\beta n \sum_{i=1}^L S_i\scr3 S_{i+1}\scr3} \prod_{i=1}^L \bigl( 1 + \tfrac{\beta h}n  S_i\scr1 \bigr) \biggr)^n \\
&= \lim_{n \to \infty} \sum_{\sigma\scr1,\dots,\sigma\scr n} \exp\biggl\{ \frac\beta{4n} \sum_{i=1}^L \sum_{k=1}^n \sigma_i\scr k \sigma_{i+1}\scr k \biggr\} \prod_{i=1}^L \prod_{k=1}^n \langle \sigma_i\scr k | \bigl( 1 + \tfrac{\beta h}n  S\scr1 \bigr) | \sigma_i\scr{k+1} \rangle.
\end{split}
\ee
We now observe that
\be
\langle \sigma | \bigl( 1 + \tfrac{\beta h}n  S\scr1 \bigr) | \sigma' \rangle = \e{-J_2\scr{n} + J_2\scr{n} \sigma \sigma'}.
\ee
Inserting this identity in \eqref{Trotter expansion} we get the proposition.
\end{proof}

Next we check that we can exchange the infinite-volume and the continuum limits for the free energy. Let us define
\be
\label{Trotter free en}
f_{L,n}^{\rm qu}(\beta,h) = -\tfrac1L \log \Tr \biggl( \e{\frac\beta n \sum_{i=1}^L S_i\scr3 S_{i+1}\scr3} \e{\frac{\beta h}n \sum_{i=1}^L S_i\scr1} \biggr)^n.
\ee
We already know that $f_L^{\rm qu}(\beta,h) = \lim_{n\to\infty} f_{L,n}^{\rm qu}(\beta,h)$ for fixed $L$.

\begin{proposition}\hfill
\label{prop limit interchange}
\begin{itemize}
\item[(a)] For fixed $n$ the limit $L\to\infty$ of $f_{L,n}^{\rm qu}(\beta,h)$ exists (and is denoted $f_{\infty,n}^{\rm qu}(\beta,h)$).
\item[(b)] We have
\[
f^{\rm qu}(\beta,h) = \lim_{L\to\infty} \lim_{n\to\infty} f_{L,n}^{\rm qu}(\beta,h) = \lim_{n\to\infty} \lim_{L\to\infty} f_{L,n}^{\rm qu}(\beta,h).
\]
\end{itemize}
\end{proposition}

\begin{proof}
Since the trace of the Lie-Trotter product can be written as a classical partition function, see Proposition \ref{prop qu-class}, we can proceed as with the usual proofs of thermodynamic limits, see \cite{FV}, and we easily obtain (a).

The first equality in (b) is clear. For the second equality we use the following estimates, which again follow from estimates on the classical partition function:
\be
Z_{L,n}^{\rm qu}(\beta,h)^k \e{-\frac{\beta k}2} \leq Z_{kL,n}^{\rm qu}(\beta,h) \leq Z_{L,n}^{\rm qu}(\beta,h)^k \e{\frac{\beta k}2}.
\ee
Taking $k\to\infty$ we get
\be
f_{L,n}^{\rm qu}(\beta,h) + \tfrac1{2L} \geq f_{\infty,n}^{\rm qu}(\beta,h) \geq f_{L,n}^{\rm qu}(\beta,h) - \tfrac1{2L}.
\ee
The rest of the proof is a standard $\frac\eps3$ argument. For any $\eps>0$ we can find $L=L(\eps)$ large enough so that for all $n$, we have
\be
\big| f^{\rm qu}(\beta,h) - f_L^{\rm qu}(\beta,h) \big| \leq \tfrac\eps3, \qquad \big| f_{\infty,n}^{\rm qu}(\beta,h) - f_{L,n}^{\rm qu}(\beta,h) \big| \leq \tfrac\eps3.
\ee
Then we can find $n_0 = n_0(\eps)$ such that $|f_L^{\rm qu}(\beta,h) - f_{L,n}^{\rm qu}(\beta,h)| \leq \frac\eps3$ for all $n \geq n_0$. Then
\be
\begin{split}
\big| f^{\rm qu}(\beta,h) &- f_{\infty,n}^{\rm qu}(\beta,h) \big| \leq \big| f^{\rm qu}(\beta,h) - f_L^{\rm qu}(\beta,h) \big| \\
&+ \big| f_L^{\rm qu}(\beta,h) - f_{L,n}^{\rm qu}(\beta,h) \big| + \big| f_{L,n}^{\rm qu}(\beta,h) - f_{\infty,n}^{\rm qu}(\beta,h) \big| \leq \eps.
\end{split}
\ee
This holds for any $\eps>0$ provided $n$ is large enough. This proves the second identity in (b).
\end{proof}

\begin{proof}[Proof of Theorem \ref{thm quantum Ising}]
We need the following identity:
\be
\label{2d to 1d}
\sum_{k_2 \in \widetilde \bbT_M} \log \bigl[ \coth(2J_2) - \cos k_2 \bigr] = -M\log2 + M \log \coth J_2 + 2 \log \bigl( 1 + (\coth J_2)^{-M} \bigr).
\ee
It can be obtained by taking the limit $J_1 \to 0$ in Theorem \ref{thm triangle Ising} (a), as the expression converges to the free energy of the 1D Ising model in $\bbT_M$. The latter is easily calculated with the 1D transfer matrices, yielding $-\log (2\cosh J_2) - \frac1M \log (1 + \tanh^M J_2 )$. We can substitute $a = \coth(2J_2)$ in the left side of Eq.\ \eqref{2d to 1d}, and $\coth J_2 = a + \sqrt{a^2-1}$ in the right side.

By Propositions \ref{prop qu-class} and \ref{prop limit interchange}, the free energy of the quantum Ising model is the limit $n\to\infty$ of
\bm
f_{\infty,n}^{\rm qu}(\beta,h) = -\tfrac n2 \log\tfrac{2\beta h}n - \tfrac n2 \log \sinh(-\log \tfrac{\beta h}{2n}) \\
- \frac1{4\pi} \int_{-\pi}^\pi \dd k_1 \sum_{k_2 \in \widetilde \bbT_n} \log \biggl[ \cosh\tfrac\beta{2n} \coth(-\log\tfrac{\beta h}{2n}) - \frac{\sinh\frac\beta{2n}}{\sinh(-\log\frac{\beta h}{2n})} \cos k_1 - \cos k_2 \biggr].
\end{multline}
We now use
\be
\begin{split}
&\cosh\tfrac\beta{2n} = 1 + \tfrac12 (\tfrac\beta{2n})^2 + O(\tfrac1{n^4}). \\
&\coth(-\log\tfrac{\beta h}{2n}) = 1 + 2(\tfrac{\beta h}{2n})^2 + O(\tfrac1{n^4}). \\
&\sinh\tfrac\beta{2n} = \tfrac\beta{2n} + O(\tfrac1{n^3}). \\
&\sinh(-\log\tfrac{\beta h}{2n}) = \tfrac12 (\tfrac{2n}{\beta h}) (1 + O(\tfrac1{n^2})).
\end{split}
\ee
Inserting in the previous expression for $f_n(\beta,h)$ we obtain
\be
f_{\infty,n}^{\rm qu}(\beta,h) = -\tfrac n2 \log 2 + O(\tfrac1n) - \frac1{4\pi} \int_{-\pi}^\pi \dd k_1 \sum_{k_2 \in \widetilde \bbT_n} \log \Bigl[ 1 + \tfrac12 (\tfrac\beta{2n})^2 \eps(h,k_1)^2 + O(\tfrac1{n^3}) - \cos k_2 \Bigr],
\ee
where we introduced
\be
\eps(h,k_1) = \sqrt{1 + 4h^2 + 4h\cos k_1}.
\ee
We now use the identity \eqref{2d to 1d} with $a = 1 + \tfrac12 (\tfrac\beta{2n})^2 \eps(h,k_1)^2 + O(\tfrac1{n^3})$, in which case we have $a + \sqrt{a^2-1} = 1 + \frac\beta{2n} \eps(h,k_1) + O(\frac1{n^2})$. We get
\be
\begin{split}
f_{\infty,n}^{\rm qu}(\beta,h) &= -\tfrac n2 \log 2 + O(\tfrac1n) - \frac1{4\pi} \int_{-\pi}^\pi \dd k_1 \Bigl\{ -n\log 2 + n \log \Bigl( 1 + \tfrac\beta{2n} \eps(h,k_1) + O(\tfrac1{n^2}) \Bigr) \\
&\hspace{5cm} + 2 \log \Bigl( 1 + \bigl( 1 + \tfrac\beta{2n} \eps(h,k_1) + O(\tfrac1{n^2}) \bigr)^{-n} \Bigr) \Bigr\} \\
&= O(\tfrac1n) - \frac1{4\pi} \int_{-\pi}^\pi \dd k_1 \Bigl\{ \tfrac\beta2 \eps(h,k_1) + 2 \log \bigl( 1 + \e{-\frac\beta2 \eps(h,k_1)} \bigr) \Bigr\}.
\end{split}
\ee
Replacing $\tfrac\beta2 \eps(h,k_1)$ by $2\log\e{\frac\beta4 \eps(h,k_1)}$ and combining the logarithms, we obtain the expression of Theorem \ref{thm quantum Ising}.
\end{proof}

We finally discuss the ``quantum phase transition" of the quantum Ising model. The free energy $f^{\rm qu}(\beta,h)$ of the one-dimensional model is clearly analytic for all $\beta>0, h\in\bbR$ (and in a complex neighbourhood), but interesting behaviour can happen in the zero-temperature limit. Namely, we consider the ground state energy
\be
e_0(h) = \lim_{\beta\to\infty} f^{\rm qu}(\beta,h).
\ee
From Theorem \ref{thm quantum Ising} we get the exact expression
\be
e_0(h) = -\frac1{8\pi} \int_{-\pi}^\pi \dd k \, \sqrt{1+4h^2+4h\cos k}.
\ee
One can check that the derivative of $e_0$ is continuous. The second derivative is
\be
e_0''(h) = -\frac1{2\pi} \int_{-\pi}^\pi \frac{\dd k}{\sqrt{1+4h^2+4h\cos k}} + \frac1{2\pi} \int_{-\pi}^\pi \dd k \frac{(2h+\cos k)^2}{(1+4h^2+4h\cos k)^{3/2}}.
\ee
The integrals are well-behaved except possibly at $h = \pm\frac12$. While the second integral has a limit as $h \to \pm\frac12$, the first integral diverges logarithmically. Precisely, we can check that
\be
e_0''(h) \sim \tfrac1{2\pi} \log |h \pm \tfrac12|
\ee
around $h=-\frac12$ and $h=\frac12$. As is well-known, there are multiple ground states when $|h| < \frac12$, that display long-range order; there is a single disordered ground state when $|h| > \frac12$. More information about the quantum Ising model can be found in the recent works \cite{Gri,BG,Iof,CI,Bjo,Li2,Tas}.

\section{Acknowledgements}

We acknowledge useful discussions with Michael Aizenman, Jakob Bj\"ornberg, David Cimasoni, S\o ren Fournais, Marcin Lis, S\'ebastien Ott, Jan Philip Solovej, Yvan Velenik, Simone Warzel, and Nikos Zygouras. We are also grateful to the referee for useful comments.
GA is supported by the EPSRC grants EP/V520226/1 and EP/W523793/1.
DU is grateful to Chalmers University in Gothenburg, and to the Mathematisches Forschungsinstitut Oberwolfach for hosting him during part of this study.

\medskip
Conflicts of interest: none.

\renewcommand{\refname}{\small References}
\bibliographystyle{symposium}

\end{document}